\documentclass[onecolumn]{IEEEtran}
\usepackage[utf8]{inputenc}
\usepackage[hmargin=1.39in]{geometry}
\usepackage{color, colortbl}
\usepackage{tikz}
\usepackage{pgfplots}
\usepackage{graphicx,amsmath}
\usepackage{subfigure}
\usepackage{amsmath,amssymb,amsfonts,url}
\usetikzlibrary{positioning,backgrounds,shadings}

\newtheorem{example}{Example}
\newtheorem{theorem}{Theorem}
\newtheorem{lemma}{Lemma}

\newtheorem{definition}{Definition}
\newtheorem{proof}{Proof}
\newtheorem{corollary}{Corollary}

\begin{document}
%\begin{frontmatter}
\title{$\prod\limits_{i=1}^{n} \mathbb{Z}_{2^i}$-Additive Cyclic Codes}
\author{Tapabrata Roy and Santanu Sarkar
\IEEEcompsocitemizethanks{\IEEEcompsocthanksitem 
 Department of Mathematics, Indian  Institute of Technology Madras,
Chennai 600 036, India.
E-mail: tapabrata.roy.048@gmail.com, sarkar.santanu.bir@gmail.com}
\thanks{}}

%\markboth{IEEE Transactions on Computers}{}

\maketitle

%\maketitle
\begin{abstract}
In this paper we study $\prod\limits_{i=1}^{n} \mathbb{Z}_{2^i}$-Additive Cyclic Codes. These codes are identified as $\mathbb{Z}_{2^n}[x]$-submodules of $\prod\limits
_{i=1}^{n}\mathbb{Z}_{2^i}[x]/ \langle x^{\alpha_i}-1\rangle$; $\alpha_i$ and $\rm{i}$ being relatively prime for each $i=1,2,\ldots,n.$ We first define a $\prod\limits_{i=1}^{n}\mathbb{Z}_{2^i}$-additive cyclic code of a certain length. We then define the distance between two codewords and the minimum distance of such a code. Moreover we relate these to binary codes using the generalized Gray maps. We define the duals of such codes and show that the dual of a $\prod\limits_{i=1}^{n}\mathbb{Z}_{2^i}$-additive cyclic code is also cyclic. We then give the polynomial definition of a $\prod\limits_{i=1}^{n}\mathbb{Z}_{2^i}$-additive cyclic code of a certain length. We then determine the structure of such codes and derive a minimal spanning set for that. We also determine the total number of codewords in this code. We finally give an illustrative example of a $\prod\limits_{i=1}^{n}\mathbb{Z}_{2^i}$-additive cyclic code.
\end{abstract}


%\end{frontmatter}

\section{Introduction}
Let $\mathbb{Z}_{2^i}=\{0,1,2,\ldots,2^{i}-1\}$ denote the ring of integers modulo $2^i$; $i=1,2,3,\ldots,n$. We follow the usual notations throughout this paper. Let $\mathbb{F}_q$ be a field with $q$ elements. A linear code over $\mathbb{F}_q$ of length $r \in \mathbb{N}$ is a subspace of the vector space $\mathbb{F}_q^{r}$ over $\mathbb{F}_q$. Later codes were defined over rings \cite{Calderbank95,Greferath99,Honold98} because of their similarities with finite fields. It is expected that the theory of linear codes over finite chain rings may resemble the same over finite fields. Here we restrict our study to $\mathbb{Z}_{2^i};~ {i}= 1,2,\ldots,n;~ n \in \mathbb{N}$. A subset $\mathcal{A} \subseteq \mathbb{Z}_2^{\alpha_1};~ \alpha_1 \in \mathbb{N}$ is called a $\mathbb{Z}_2$-linear code or a binary linear code if it is a $\mathbb{Z}_2$-subspace of  $\mathbb{Z}_2^{\alpha_1}$. For $a=(a_1,a_2,\ldots,a_{\alpha_1}) \in \mathbb{Z}_2^{\alpha_1}$ the Hamming weight of $a$ is defined to be the number of non-zero coordinates in $a$ i.e. $$wt_H(a)=\mid \{i \in \{1,2,\ldots,\alpha_1\}: a_i \neq 0\} \mid .$$ For $a=(a_1,a_2,\ldots,a_{\alpha_1})$ and $b=(b_1,b_2,\ldots,b_{\alpha_1})$ the Hamming distance between $a$ and $b$ is defined to be the number of coordinates where $a$ and $b$ differ, i.e. $$d_H(a,b)=wt_H(a-b).$$ A subset $\mathcal{C} \subseteq \mathbb{Z}_{2^i}^{\alpha_i}$ is called a $\mathbb{Z}_{2^i}$-code if it is a $\mathbb{Z}_{2^i}$-submodule of $\mathbb{Z}_{2^i}^{\alpha_i}; $ $\alpha_i \in \mathbb{N},~ i \in \{1,2,3,\ldots,n\}$. The Lee weight of $a=(a_1,a_2,\ldots,a_{\alpha_i}) \in \mathbb{Z}_{2^i}^{\alpha_i}$ is denoted by $wt_L(a)$ and is defined by $$wt_L(a)=\sum\limits_{j_i=1}^{\alpha_i} min\{a_{j_i},2^i-a_{j_i}\}.$$ The Lee distance between $a=(a_1,a_2,\ldots,a_{\alpha_i})$ and $b=(b_1,b_2,\ldots,b_{\alpha_i}) \in  \mathbb{Z}_{2^i}^{\alpha_i}$ is defined by $$d_L(a,b)=\sum\limits_{j_i=1}^{\alpha_i} min\{\mid a_{j_i}-b_{j_i}\mid, 2^i- \mid a_{j_i}-b_{j_i}\mid\}.$$ From the very beginning binary cyclic codes are treated as one of the most important classes of codes \cite{Pless89}. Sloane et al. showed in 1994 that binary codes can be found as images of linear codes over $\mathbb{Z}_4$ under a non linear Gray map \cite{Hammons94}. This grew interest in many researchers to study codes over various rings. Since 1990's many papers \cite{Abualrub04,Abualrub07,Abualrub03,Hammons94,Pless96} studied codes over $\mathbb{Z}_4$. In \cite{Pless96} authors characterized the ideals in $\mathbb{Z}_4[x]/\langle f(x) \rangle; f(x)$ being a basic irreducible polynomial over $\mathbb{Z}_4$, and found $\mathbb{Z}_4$-cyclic codes. In \cite{Roy17} authors generalized these ideas to obtain the ideals in $\mathbb{Z}_{2^n}[x]/\langle f(x) \rangle; f(x)$ being a basic irreducible polynomial in $\mathbb{Z}_{2^n}[x]$. They also found $\mathbb{Z}_{2^n}[x]$-cyclic codes. Recently $\mathbb{Z}_2 \mathbb{Z}_4$-additive cyclic codes have emerged \cite{Abualrub14,Bilal11,Borges09,Borges16,Rifa11}. A $\mathbb{Z}_2\mathbb{Z}_4$-additive code can be identified as a $\mathbb{Z}_4[x]$-submodule of $R_{r,s}=\mathbb{Z}_2[x]/\langle x^r-1 \rangle \times \mathbb{Z}_4[x]/\langle x^s-1 \rangle$. They determined the structure of such codes and also obtained a minimal spanning set for that. In \cite{Aydogdu13,Aydogdu14,Borges15,Roy17} authors generalized these ideas. These codes have shown promising applications to many areas. Such as, the concepts in \cite{Borges09} have been implemented by authors in \cite{Borges07}, perfect $\mathbb{Z}_2 \mathbb{Z}_4$-additive cyclic codes have been used in steganography \cite{Rifa11}. Here we study $\prod \limits_{i=1}^{n} \mathbb{Z}_{2^i}$-additive cyclic codes. This we do as an extension of \cite{Abualrub14} and \cite{Aydogdu17}. 

We begin with the definition of a $\prod \limits_{i=1}^{n} \mathbb{Z}_{2^i}$-additive code and a $\prod \limits_{i=1}^{n} \mathbb{Z}_{2^i}$-additive cyclic code. Then we give a notion of distance in it. We also relate these to binary codes by defining generalized Gray maps. We then identify these codes as $\mathbb{Z}_{2^n}[x]$-submodules of $\prod \limits_{i=1}^{n} \mathbb{Z}_{2^i}[x]/\langle x^{\alpha_i}-1 \rangle; ~gcd(i, \alpha_i)=1 ~\forall ~i=1,2,\ldots,n$, and give the polynomial definition of such codes. We then determine the generator polynomials, a minimal spanning set and the total number of codewords for a $\prod \limits_{i=1}^{n} \mathbb{Z}_{2^i}[x]$-additive cyclic code. Finally we present an illustrative example.

\section{$\prod\limits_{i=1}^{n} \mathbb{Z}_{2^i}$-Cyclic Codes}

Consider the set 
$$\prod_{i=1}^n \mathbb{Z}_{2^i}=\mathbb{Z}_2 \times \mathbb{Z}_4 \times \cdots \times \mathbb{Z}_{2^n}=\{(u_1, u_2, \ldots, u_n) \ | \ u_i \in \mathbb{Z}_{2^i} ~ \forall ~i=1,2,\ldots,n\}. $$
This set is closed under addition. We make it closed under multiplication by elements of $\mathbb{Z}_{2^n}$ by defining the multiplication 
$$v \cdot (u_1,u_2,\ldots, u_n)=(vu_1 \bmod 2, vu_2 \bmod 4, \ldots, vu_n \bmod 2^n ),$$
$\forall~ v \in \mathbb{Z}_{2^n}, (u_1,u_2,\ldots, u_n) \in \prod\limits_{i=1}^n \mathbb{Z}_{2^i}.$ Defining this multiplication 
$ \prod\limits_{i=1}^n \mathbb{Z}_{2^i}$ becomes a $\mathbb{Z}_{2^n}$-module. 

\begin{definition}
A subgroup, $\mathcal{C}$, of $\prod\limits_{i=1}^n \mathbb{Z}^{\alpha_i}_{2^i}; \alpha_i \in \mathbb{Z}^{\#}$ will be called a  
$\prod\limits_{i=1}^n \mathbb{Z}_{2^i}$-additive code.

\end{definition}
for simplicity we consider that for such a code $\mathcal{C}$, first $\alpha_1$ coordinates consists of elements from $\mathbb{Z}_2$,
next $\alpha_2$ coordinates consists of elements from $\mathbb{Z}_4$ and so on until the last $\alpha_n$
coordinates consists of elements from $\mathbb{Z}_{2^n}.$ We denote this by 
$$(v_1,v_2, \ldots, v_n) \in \mathcal{C}, $$
where $v_i=(v_{i1}, v_{i2}, \ldots, v_{i \alpha_i}) \in \mathbb{Z}^{\alpha_i}_{2^i}$ is a $\alpha_i$-tuple.
For some well known results, we will assume that $\gcd(i, \alpha_i)=1$ $\forall ~ i \in \{1,2,\ldots,n\}$ throughout the paper. Clearly any subgroup of 
$\prod\limits_{i=1}^n \mathbb{Z}^{\alpha_i}_{2^i}$ is of the form $\prod\limits_{i=1}^n \mathbb{Z}^{\beta_i}_{2^i}; 0 \leq \beta_i \leq \alpha_i$.
So taking this into account, we define the following.

\begin{definition}
A subgroup $\mathcal{C} \cong \prod\limits_{i=1}^n \mathbb{Z}^{\beta_i}_{2^i}$ of $\prod\limits_{i=1}^n \mathbb{Z}^{\alpha_i}_{2^i}$ is called 
a $\prod\limits_{i=1}^n \mathbb{Z}_{2^i}$-additive code of type $(\alpha_1, \alpha_2, $ $\ldots,$ $ \alpha_n, \beta_1, \beta_2, \ldots, \beta_n).$
\end{definition}

We  now measure distance between two codewords in a code of above type. We define the Gray map for $q \in \mathbb{N}$.
Define a map $\phi: \mathbb{Z}_{2q}\longrightarrow \mathbb{Z}_2^q$ by $\phi(0)=(0,0,0,\ldots,0,0,0)$, $\phi(kq+i)+\phi(kq+i+1)=(0,0,\ldots,1,0,\ldots,0)$, where $(q-i)$th entry is $1$ and all other entries are $0$ for all 
$i=0,1,\ldots,q-1$ and $k=0,1$; for example, take $q=4$ and define $\phi: \mathbb{Z}_{8}\longrightarrow \mathbb{Z}_2^4$ by 
$\phi(0)=(0,0,0,0)$, $\phi(1)=(0,0,0,1)$, $\phi(2)=(0,0,1,1)$, $\phi(3)=(0,1,1,1)$, $\phi(4)=(1,1,1,1)$, $\phi(5)=(1,1,1,0)$, $\phi(6)=(1,1,0,0)$, $\phi(7)=(1,0,0,0)$. Let $\phi_i: \mathbb{Z}_{2^i} \rightarrow \mathbb{Z}^{2^{i-1}}_2, 
$ be the Gray maps defined above. So, we can define a Gray map
$$\Phi: \prod_{i=1}^n \mathbb{Z}^{\alpha_i}_{2^i} \rightarrow \mathbb{Z}^{s}_2 ;~ s= \sum\limits_{j=1}^n 2^{j-1}\alpha_j, $$ by 
$\Phi(v_{11},v_{12}, \ldots, v_{1\alpha_1},v_{21},v_{22}, \ldots, v_{2\alpha_2}, \ldots, v_{n1},v_{n2},\ldots, v_{n\alpha_n}) \\ =(v_{11},v_{12}, \ldots, v_{1\alpha_1},$ $\phi_2(v_{21}),\phi_2(v_{22}), \ldots,\phi_2(v_{2\alpha_2}),\ldots,  \phi_n(v_{n1}),\phi_n(v_{n2}),\ldots, 
\phi_n(v_{n\alpha_n})).$ 

Clearly this map $\Phi$ transforms the Lee distance in $\prod\limits_{i=1}^n \mathbb{Z}^{\alpha_i}_{2^i}$ into the Hamming distance in
$\mathbb{Z}^{s}_2;~s=\sum\limits_{j=1}^n 2^{j-1}\alpha_j.$
For any $(v_1,v_2, \ldots, v_n) \in \prod\limits_{i=1}^n\mathbb{Z}^{\alpha_i}_{2^i}$ with $v_i=(v_{i1},v_{i2},\ldots,v_{i\alpha_i}) \in \mathbb{Z}^{\alpha_i}_{2^i}$ for $1 \leq i \leq n$, the weight of $(v_1,v_2,\ldots,v_n)$ is defined to be
$$wt(v_1,v_2,\ldots, v_n)=wt_H(v_1)+\displaystyle\sum_{i=2}^n wt_{L_i}(v_i); $$ $wt_H(v_1)$ being the Hamming weight of $v_1$ in $\mathbb{Z}^{\alpha_1}_2$
and $wt_{L_i}(v_i)$ being the Lee weight of $v_i$ in $\mathbb{Z}^{\alpha_i}_{2^i}; ~ i=2,3,\ldots,n.$ For two element 
$c_1,c_2 \in \prod\limits_{i=1}^n\mathbb{Z}^{\alpha_i}_{2^i} $ the distance between $c_1$ and $c_2$ is denoted by $d(c_1,c_2)$ and is defined 
to be $$d(c_1,c_2)=wt(c_1-c_2).$$ For a code $\mathcal{C}$, we denote the minimum distance by $d(\mathcal{C})$ and define it as $$d(\mathcal{C})=\min_{c_1,c_2 \in \mathcal{C}}\{d(c_1,c_2); c_1 \neq c_2\}.$$ 
The binary image $\Phi(\mathcal{C})$ of $\mathcal{C}$ of the above type is a binary code of length $s=\displaystyle\sum_{j=1}^n 2^{j-1}\alpha_j$ and 
size $2^{t};~t=\sum\limits_{j=1}^n j\beta_j$. It will be called a $\prod\limits_{i=1}^n \mathbb{Z}_{2^i}$-linear code. Now extending the classical definition
of a cyclic code naturally, we define a $\prod\limits_{i=1}^n \mathbb{Z}_{2^i}$-additive cyclic code.

\begin{definition}
A subset $\mathcal{C}$ of $\prod\limits_{i=1}^n \mathbb{Z}^{\alpha_i}_{2^i}$ is called a $\prod\limits_{i=1}^n \mathbb{Z}_{2^i}$-additive cyclic code if
\begin{enumerate}
\item $\mathcal{C}$ is an additive code and
\item For any code word $v=(v_{11},v_{12},\ldots, v_{1\alpha_1},v_{21},v_{22},\ldots, v_{2\alpha_2},v_{n1},v_{n2},\ldots,v_{n\alpha_n}) \in \mathcal{C}$
its cyclic shift,

$T(u)= (v_{1\alpha_1},v_{11},v_{12},\ldots, v_{1(\alpha_1-1)},v_{2\alpha_2},v_{21},v_{22},\ldots, v_{2(\alpha_2-1)},  \\
 ~~~~~~~~~~~~~~~~~~~~~~~~~~~~~~~~~~~~~~~~~~~ \ldots,v_{n\alpha_n},
v_{n1},v_{n2},\ldots,v_{n(\alpha_n-1)}),$
is also in $\mathcal{C}$.
\end{enumerate}
\end{definition}
To define the dual of such a code we define the following inner product in $\prod\limits_{i=1}^n\mathbb{Z}^{\alpha_i}_{2^i}.$

Let $u=(u_1,u_2,\ldots,u_n)$ and $v=(v_1,v_2,\ldots,v_n)$ be two elements in $\prod\limits_{i=1}^n \mathbb{Z}^{\alpha_i}_{2^i}$
where $u_i=(u_{i1},u_{i2},\ldots, u_{i\alpha_i}),$ $v_i=(v_{i1},v_{i2},\ldots, v_{i\alpha_i}) \in \mathbb{Z}^{\alpha_i}_{2^i} ~ \forall ~ i \in \{1,2,\ldots,n\}.$
Define the inner product 
$$u\cdot v=\bigg[\displaystyle\sum_{i=1}^n \big(2^{n-i}\displaystyle\sum_{j_i=1}^{\alpha_i}u_{ij_i}v_{ij_i}\big)\bigg] \bmod 2^n. $$

\begin{definition}
Let $\mathcal{C}$ be any $\prod\limits_{i=1}^n \mathbb{Z}^{\alpha_i}_{2^i}$ additive cyclic code. The dual of $\mathcal{C}$ is denoted by $\mathcal{C}^{\bot}$
and is defined to be $\mathcal{C}^{\bot}=\{v \in \prod\limits_{i=1}^n \mathbb{Z}^{\alpha_i}_{2^i} \ | \ u \cdot v=0 \ \forall \ u \in \mathcal{C}\}.$
Clearly the dual of a code $\mathcal{C}$ is also an additive code.

\end{definition} 

We have the following lemma.

\begin{lemma}
The dual of a $\mathbb{Z}^{\alpha_i}_{2^i}$-additive cyclic code. Then the dual $\mathcal{C}^{\bot}$ of $\mathcal{C}$ is also cyclic.
\end{lemma}
\begin{proof}

Let $u=(u_{11},\ldots,u_{1\alpha_1},u_{21},\ldots,u_{2\alpha_2}, \ldots,u_{n1},\ldots,u_{n\alpha_n}) \in \mathcal{C}^{\bot}.$ It is enough to show that $T(u) \in \mathcal{C}^{\bot}.$ We have $u \in \mathcal{C}^{\bot}$. So, if we take any element $v$ from $\mathcal{C}$ with 
$$v=(v_{11},\ldots,v_{1\alpha_1},v_{21},\ldots,v_{2\alpha_2},\ldots,v_{n1},\ldots,v_{n\alpha_n}),$$ we must have 
$$u\cdot v=\bigg[\displaystyle\sum_{i=1}^n \big(2^{n-i}\displaystyle\sum_{j_i=1}^{\alpha_i}u_{ij_i}v_{ij_i}\big)\bigg]=0 \bmod 2^n. $$
Let $k=\mbox{lcm}(\alpha_1,\alpha_2,\ldots, \alpha_n).$ Then $T^k(v)=v \ \forall \ v \in \prod\limits_{i=1}^n\mathbb{Z}^{\alpha_i}_{2^i}.$ 

Let $w=T^{k-1}(v)=(v_{12} \ldots v_{1\alpha_1}v_{11},v_{22}$ $\ldots v_{2\alpha_2}v_{21},\ldots, v_{n2}\ldots v_{n\alpha_n}v_{n1}).$
Now $w \in \mathcal{C}$. So, 
$$\begin{array}{ll}
0 &=w \cdot u \\ &= \bigg[2^{n-1}(v_{11}u_{1\alpha_1}+v_{12}u_{11}+\cdots+v_{1\alpha_1}u_{1(\alpha_1-1)})\\
&~~~~~~~~~~~~~~~~~~~~~~~  +2^{n-2}(v_{21}u_{2\alpha_2}+v_{22}u_{21}+ \cdots+v_{2\alpha_2} u_{2(\alpha_2 -1)})+ \\
&~~~~~~~~~~~~~~~~~~~~~~~~~~~~~~~~ \cdots+(v_{n1}u_{n\alpha_n}+v_{n2}u_{n1}+\cdots+v_{n\alpha_n}u_{n(\alpha_n-1)
})\bigg] \bmod 2^n \\ &=v \cdot T(u).
\end{array}$$

Thus $v \cdot T(u)=0 \\ \Rightarrow T(u) \in \mathcal{C}^{\bot}$ and hence $\mathcal{C}^{\bot}$ is cyclic.
%\qed
\end{proof}
Denote $\mathcal{R}_{\alpha_1,\alpha_2, \ldots, \alpha_n}=\prod\limits_{i=1}^n \mathbb{Z}_{2^i}[x]/\langle x^{\alpha_i}-1\rangle.$

Furthermore any element $u=(u_{10},u_{11},\ldots, u_{1(\alpha_1-1)},u_{20},u_{21},\ldots,$ $u_{2(\alpha_2-1)},\ldots,u_{n0},$  $u_{n1},$ $\ldots,$ $ u_{n(\alpha_n-1)})=(u_1,u_2,\ldots,u_n); u_i=(u_{i0},u_{i1}, \ldots, u_{i(\alpha_i-1)}) \in \mathbb{Z}^{\alpha_i}_{2^i}$ for $1 \leq i \leq n$
can be identified as

 $u(x)=\bigg(\displaystyle\sum_{j_1=0}^{\alpha_1-1} u_{1j_1}x^{j_1 }, \displaystyle\sum_{j_2=0}^{\alpha_2-1} u_{2j_2}x^{j_2},\ldots, 
\displaystyle\sum_{j_n=0}^{\alpha_n-1} u_{nj_n}x^{j_n}\bigg)=(u_1(x),u_2(x), \ldots,u_n(x));$

$u_i(x)=\displaystyle\sum_{j_i=0}^{\alpha_i-1} u_{ij_i}x^{j_i}$ for $1 \leq i \leq n$ in 
$\mathcal{R}_{\alpha_1,\alpha_2,\ldots,\alpha_n}.$ 

Clearly this is a one to one correspondence between the elements in $\prod\limits_{i=1}^n\mathbb{Z}^{\alpha_i}_{2^i}$ and $\mathcal{R}_{\alpha_1,\alpha_2,\ldots, \alpha_n}$. Define a scalar multiplication for $d(x)\in \mathbb{Z}_{2^n}[x]$ and 
$(u_1(x), \ldots, u_n(x))$ $ \in \mathcal{R}_{\alpha_1,\alpha_2,\ldots, \alpha_n}$ by 
$d(x)*(u_1(x), u_2(x), \ldots, u_n(x))=(d(x)u_1(x) \bmod 2, d(x)u_2(x) \bmod 4,$ $ \ldots, d(x)u_n(x) \bmod 2^n).$
Clearly this multiplication is well defined and $\mathcal{R}_{\alpha_1,\ldots, \alpha_n}$ is a $\mathbb{Z}_{2^n}[x]$-module
with respect to this multiplication. Now we give the polynomial definition of a $\prod\limits_{i=1}^n \mathbb{Z}^{\alpha_i}_{2^i}$-additive cyclic code which is obviously equivalent to what we defined earlier.
\begin{definition}
A subset $\mathcal{C} \subseteq \mathcal{R}_{\alpha_1, \alpha_2,\ldots,\alpha_n}$ is defined to be a $\prod\limits_{i=1}^n \mathbb{Z}^{\alpha_i}_{2^i}$-additive cyclic code if $\mathcal{C}$ is a subgroup of $\mathcal{R}_{\alpha_1,\ldots,\alpha_n}$ and for all $u(x)=(u_1(x), u_2(x), \ldots, u_n(x)) \in \mathcal{C},$
$x * u(x)$ is in $\mathcal{C}$. 
\end{definition}
From the above the following theorem follows easily. 

\begin{theorem}
A code $\mathcal{C}$ is a $\prod\limits_{i=1}^n \mathbb{Z}^{\alpha_i}_{2^i}$-additive cyclic code iff $\mathcal{C}$ is a $\mathbb{Z}_{2^n}[x]$-submodule 
of $$\mathcal{R}_{\alpha_1,\alpha_2,\ldots, \alpha_n}= \prod\limits_{i=1}^n \mathbb{Z}_{2^i}[x]/\langle x^{\alpha_i}-1\rangle.$$
\end{theorem}

\section{Structure of $\prod\limits_{i=1}^n \mathbb{Z}_{2^i}$-additive cyclic codes}
In this section, we find a set of generators for a $\prod\limits_{i=1}^n \mathbb{Z}^{\alpha_i}_{2^i}$ additive cyclic code $\mathcal{C}$, as a $\mathbb{Z}_{2^n}[x]$-submodule of $\mathcal{R}_{\alpha_1,\alpha_2,\ldots,\alpha_n}.$ Now $\mathcal{C}$ and $\mathbb{Z}_{2^n}[x]/\langle x^{\alpha_n}-1\rangle$ are both
$\mathbb{Z}_{2^n}[x]$-submodules of $\mathcal{R}_{\alpha_1,\alpha_2,\ldots,\alpha_n}$. Define
$$\Psi: \mathcal{C} \rightarrow \mathbb{Z}_{2^n}[x]/\langle x^{\alpha_n}-1\rangle $$
by $\Psi(u_1(x),u_2(x),\ldots,u_n(x))=u_n(x).$\\
As we now treat codewords as polynomials, so for simplicity we write $f$ instead of $f(x)$ for all polynomials $f(x)$. So,
$\Psi(u_1,u_2,\ldots,u_n)=u_n.$
It is easy to check that $\Psi$ is a module homomorphism and $Im(\Psi)$ is an ideal in $\mathbb{Z}_{2^n}[x]/ \langle x^{\alpha_n}-1\rangle$ and also a $\mathbb{Z}_{2^n}[x]$-submodule of $\mathbb{Z}_{2^n}[x]/\langle x^{\alpha_n}-1\rangle$. So, we have from \cite{Roy17},
$$Im(\Psi)=\bigg\langle \displaystyle\sum_{j_n=0}^{n-1}2^{j_n}a_{nj_n}(x)\bigg\rangle=\langle a_{n0},2a_{n1},\ldots, 2^{n-1}a_{n(n-1)}\rangle $$
and $\ker(\Psi)=\{(u_1,\ldots,u_{n-1},0) \in \mathcal{C} \ |\ u_i \in \mathbb{Z}_{2^i}[x]/\langle x^{\alpha_i}-1\rangle  ~\forall ~ 
i=1,2,\ldots,n-1$\}.
So let $$I=\langle \{(u_1,u_2,\ldots,u_{n-1})\in \mathcal{R}_{\alpha_1,\ldots, \alpha_{n-1}} \ \mid \ (u_1,u_2,\ldots,u_{n-1},0) \in \mathcal{C}\}\rangle.$$

Obviously $I$ is a $\prod\limits_{i=1}^{n-1}\mathbb{Z}_{2^i}$-additive cyclic code. The following theorem gives the idea of the structure of a 
$\prod\limits_{i=1}^n \mathbb{Z}_{2^i}$-additive cyclic code $\mathcal{C}$.

\begin{theorem}\label{Th2}
A $\prod\limits_{i=1}^n \mathbb{Z}_{2^i}$-additive cyclic code $\mathcal{C}$ is generated by \\
$\mathcal{C}=\langle (a_1,0,0, \ldots,0),(l_{21},a_2,0,0,\ldots,0), (l_{31},l_{32},a_3,0,\ldots,0),\\ ~~~~~~~~~~~~~~~~~~~~~~~~~~~~~~~~~~~~~~ ~~~~~~~~~~~~~~~~~~~~~~~~~~~~~~~~~~~~~~~  \ldots, (l_{n1},l_{n2}, \ldots,l_{n(n-1)},a_n)\rangle,$ \\
where each generator, $(l_{i1},l_{i2},\ldots,\l_{i(i-1)},a_i,0,\ldots,0);~i=1,2,\ldots,n,$ is an $n$-tuple and $a_i=\displaystyle\sum_{j_i=0}^{i-1}2^{j_i}a_{ij_i}$ and $l_{i j_i};~ j_i=1,2,\ldots,i-1; i=1,2,\ldots,n$ are polynomials satisfying 
\begin{enumerate}
\item[(i)] $a_{i(i-1)}| a_{i(i-2)}| \cdots | a_{i0} | x^{\alpha_i}-1 \bmod 2^i \ \forall \ i=1,2,\ldots,n$
\item[(ii)] $\deg l_{(i+1)1}< \deg a_1$ and $\deg l_{(i+1)i}< \deg a_{i0} \ \forall \ i=1,2,\ldots,n-1$
\item[(iii)] $a_i | \frac{x^{\alpha_{i+1}}-1}{a_{(i+1)i}}l_{(i+1)i} \bmod 2^i \ \forall \ i=1,2,\ldots,n-1$
\item[(iv)] $a_{i-1} | d_il_{i(i-1)}-\frac{x^{\alpha_{i+1}}-1}{a_{(i+1)i}}l_{(i+1)(i-1)} \bmod 2^{i-1} \ \forall \ i=2,3,\ldots,n-1,$\\
where $a_id_i=\frac{x^{\alpha_{i+1}}-1}{a_{(i+1)i}}l_{(i+1)i} \bmod 2^i~\forall~i=1,2,\ldots,n-1.$

\end{enumerate}
\end{theorem}

\begin{proof}
We prove this by applying the method of induction. For $n=1$, the case is similar to linear binary cyclic codes and hence, by well known theorems on binary cyclic codes, is true. For $n=2$, we know from \cite{Abualrub14} that 
$$\mathcal{C}=\langle (a_{10},0),(l_{21}, a_{20}+2a_{21})\rangle, $$
where $a_{10} | x^{\alpha_1}-1 \bmod 2$,
     $a_{21}|a_{20}|x^{\alpha_2}-1 \bmod 4$,
     $\deg l_{21}< \deg a_{10}$ and
     $a_{10}| \frac{x^{\alpha_2}-1}{a_{21}} l_{21} \bmod 2$

Hence the statement is true for $n=2$. Now let us assume that the given statement holds for $n=m-1$. We will prove it to be true for $n=m$.
Let $\mathcal{C}$ be a $\prod\limits_{i=1}^m \mathbb{Z}_{2^i}$-additive cyclic code. Consider the map 
$$\Psi: \mathcal{C} \rightarrow \mathbb{Z}_{2^m}[x]/ \langle x^{\alpha_m}-1 \rangle. $$

From \cite{Roy17} we have
$Im(\Psi)=\bigg\langle \displaystyle\sum_{j_m=0}^{m-1}2^{j_m}a_{mj_m}(x) \bigg\rangle$ where 
\begin{equation}
\label{eq1}
a_{m(m-1)} | a_{m(m-2)}| \cdots |a_{m0} | x^{\alpha_m}-1 \bmod 2^m
\end{equation}
and $\ker(\Psi)=\{(u_1,u_2, \ldots, u_{m-1},0) \in \mathcal{C} \ | \ u_i \in \mathbb{Z}_{2^i}[x]/\langle x^{\alpha_i}-1\rangle \ \forall \ i=1,2,\ldots, m-1\}$
Let $I=\{(u_1,u_2,\ldots, u_{m-1}) \in \mathcal{R}_{\alpha_1,\alpha_2,\ldots, \alpha_{m-1}} \ | \ (u_1,u_2,\ldots, u_{m-1},0) \in \mathcal{C}\}.$

Clearly $I$ is an ideal and obviously is a $\prod\limits_{i=1}^{m-1}\mathbb{Z}_{2^i}$-additive cyclic code. 
So, by the induction hypothesis,\\ $I=\langle (a_1,0,\ldots,0), (l_{21},a_2,0,\ldots,0), (l_{31},l_{32},a_3,0,\ldots,0), \ldots, \\ ~~~~~~~~~~~~~~~~~~~~~~~~~~~~~~~~~ ~~~~~~~~~~~~ (l_{(m-1)1}, l_{(m-1)2}
, \ldots, l_{(m-1)(m-2)},a_{m-1})\rangle,$ \\ where each generator, $(l_{i1},l_{i2},\ldots,\l_{i(i-1)},a_i,0,\ldots,0);~i=1,2,\ldots,m-1,$ is an $(m-1)$-tuple and $a_i=\displaystyle\sum_{j_i=0}^{i-1}2^{j_i}a_{ij_i}$ and $l_{i j_i};~ j_i=1,2,\ldots,i-1; i=1,2,\ldots,m-1$ are polynomials satisfying 
\begin{enumerate}
\item[(a)] $a_{i(i-1)}| a_{i(i-2)}| \cdots | a_{i0} | x^{\alpha_i}-1 \bmod 2^i \ \forall \ i=1,2,\ldots, m-1$
\item[(b)] $\deg l_{(i+1)1}<  \deg a_1, \deg l_{(i+1)i} < \deg a_{i0} \ \forall \ i=1,2,\ldots, m-2$ 
\item[(c)] $a_i | \frac{x^{\alpha_{i+1}}-1}{a_{(i+1)i}} l_{(i+1)i} \bmod 2^i \ \forall \ i=1,2,\ldots, m-2$
\item[(d)] $a_{i-1} | d_il_{i(i-1)}-\frac{x^{\alpha_{i+1}}-1}{a_{(i+1)i}} l_{(i+1)(i-1)} \bmod 2^{i-1} \ \forall \ i=2,3,\ldots, m-2$.
\end{enumerate}
For any, $(u_1,u_2, \ldots, u_{m-1},0) \in \ker \Psi, \exists ~ d_i(x) \in \mathbb{Z}_{2^m}[x];~ i \in \{1,2,\ldots,m-1\}$ such that
$$\begin{array}{ll}
(u_1,u_2, \ldots, u_{m-1},0) &=d_1 * (a_1,0,\ldots,0)+d_2*(l_{21},a_2,0,\ldots,0)+d_3*(l_{31},l_{32},a_3,0,\ldots,0)\\
 &~~~~~~~~~~~~~
+\cdots+ d_{m-1}*(l_{(m-1)1},l_{(m-1)2}, \ldots, l_{(m-1)(m-2)},a_{m-1},0).
\end{array}$$

So, we may assume that $\ker(\Psi)$ is a $\mathbb{Z}_{2^m}[x]$-submodule of $\mathcal{C}$ generated by 

$\langle (a_1,0,\ldots,0), (l_{21},a_2,0,\ldots,0), \ldots, (l_{(m-1)1}, l_{(m-1)2},\ldots, l_{(m-1)(m-2)},a_{m-1},0)\rangle $, each being an $m$-tuple.

Again by the First Isomorphism theorem, $$\mathcal{C}/\ker \Psi \cong Im (\Psi) ,$$
i.e, 
$$\mathcal{C}/\ker \Psi \cong \bigg\langle\displaystyle\sum_{j_m=0}^{m-1}2^{j_m}a_{mj_m} \bigg\rangle.$$
Now, let $(l_{m1},l_{m2},\ldots,l_{m(m-1)},a_m) \in \mathcal{C}$ be such that 
$$\Psi(\langle l_{m1},l_{m2}, \ldots,l_{m(m-1)},a_m\rangle)=\langle a_m\rangle,$$
where $a_m=\displaystyle\sum_{j_m=0}^{m-1}2^{j_m}a_{mj_m}$.

Hence a $\prod\limits_{i=1}^m \mathbb{Z}_{2^i}$-additive cyclic code can be generated by $m$ elements of the form
 $(a_1,$ $0,\ldots,0),$ $ (l_{21},a_2,0,\ldots,0),\ldots, (l_{(m-1)1},l_{(m-1)2}, \ldots, l_{(m-1)(m-2)},a_{m-1},0),$  
$(l_{m1},$ $l_{m2},$ $ \ldots, l_{m(m-1)},a_m)$. 

So, any element in $\mathcal{C}$ can be represented as $$d_1*(a_1,0,\ldots,0)+d_2*(l_{21},a_2,0,\ldots,0) +\cdots+d_m*(l_{m1},l_{m2}, \ldots,l_{m(m-1)},a_m).$$ So, $$\mathcal{C}=\langle (a_1,0,\ldots,0),(l_{21},a_2,0,\ldots,0), (l_{31},l_{32},a_3,0,\ldots,0)
 , \ldots, (l_{m1},l_{m2}, \ldots, l_{m(m-1)},a_m) \rangle$$
each generator being an $m$-tuple and $a_i=\sum\limits_{j_i=0}^{i-1}2^{j_i}a_{ij_i} \ \forall \ i=1,2,\ldots,m.$
\begin{enumerate}
\item[(i)] Clearly (a) together with~\eqref{eq1} proves (i).

\item[(ii)] Let $\deg l_{m_1} \geq \deg a_1$ with $\deg l_{m1} - \deg a_1=k$. Now let 
$\mathcal{D}$ be a $\prod\limits_{i=1}^m \mathbb{Z}_{2^i}$-additive cyclic code with generators \\
$~~~~(a_1,0,\ldots,0), (l_{21},a_2,0,\ldots,0), \ldots, (l_{(m-1)1},l_{(m-1)2}, \ldots, l_{(m-1)(m-2)},a_{m-1},0), \\ ~~~~~~~~~~~~~~~~~~~~~~~~~~~~~~~~~~~~~~~~~~~~~~~~~~~~ ~~~~~~~~~~~~~~~~~~~ (l_{m1}+x^ka_1,l_{m2}, \ldots, l_{m(m-1)},a_m) \\ =(a_1,0,\ldots,0), (l_{21},a_2,0,\ldots,0), \ldots, (l_{(m-1)1},l_{(m-1)2}, \ldots, l_{(m-1)(m-2)},a_{m-1},0), \\ ~~~~~~~~~~~~~~~~~~~~~~~~~~~~~~~~~~~~~~~~~~~~~~~~~~~~ ~~~~~~~~~~~~~~~~~~~~ ((l_{m1},l_{m2}, \ldots, l_{m(m-1)},a_m)+x^k*(a_1,0,\ldots,0)).$ So $\mathcal{D} \subseteq \mathcal{C}$. Again $(l_{m1},l_{m2},\ldots,l_{m(m-1)},a_m)=(l_{m1}+x^ ka_1, 
l_{m2},\ldots, l_{m(m-1)},a_m)-x^i*(a_1,0,\ldots,0)$. Hence $\mathcal{C}=\mathcal{D}$. 
Thus we may assume 
\begin{equation}
\label{eq2}
\deg l_{m1} < \deg a_1
\end{equation}

So, (b) together with~\eqref{eq2} proves the first statement of (ii). Others may also proved similarly.

\item[(iii)] Clearly \\
$~~~~~~~~~~~~~\frac{x^{\alpha_m}-1}{a_{m(m-1)}}*(l_{m1},l_{m2},\ldots,l_{m(m-1)},a_m) \\ ~~~~~~~~~~~~~~~~~~~~~~~~~~~~~~~~~~~~~~~~~~~~~~ =(\frac{x^{\alpha_m}-1}{a_{m(m-1)}}l_{m1},\frac{x^{\alpha_m}-1}{a_{m(m-1)}}l_{m2}, \ldots, \frac{x^{\alpha_m}-1}{a_{m(m-1)}}l_{m(m-1)},0 ).$

Since \\
$~~~~~~~~~~~~~~\Psi(\frac{x^{\alpha_m}-1}{a_{m(m-1)}}l_{m1},\frac{x^{\alpha_m}-1}{a_{m(m-1)}}l_{m2}, \ldots, \frac{x^{\alpha_m}-1}{a_{m(m-1)}}l_{m(m-1)},0 )=0, \\~~~~~~~~~~~~~~~~~~~~~~~~~~~~~~~~~~~~~~~~~~~~~~~ $
$(\frac{x^{\alpha_m}-1}{a_{m(m-1)}}l_{m1},\frac{x^{\alpha_m}-1}{a_{m(m-1)}}l_{m2}, \ldots, \frac{x^{\alpha_m}-1}{a_{m(m-1)}}
l_{m(m-1)},0) \in \ker \Psi$. 

That is \\ $~~~~(\frac{x^{\alpha_m}-1}{a_{m(m-1)}}l_{m1},\frac{x^{\alpha_m}-1}{a_{m(m-1)}}l_{m2}, \ldots, \frac{x^{\alpha_m}-1}{a_{m(m-1)}}
l_{m(m-1)},0)\\ =e_1(x)*(a_1,0,\ldots,0)+e_2(x)*(l_{21},a_2,0,\ldots,0)+ \cdots+ \\ ~~~~~~~~~~~~~~~~~~~~~~~~~~~~~~~e_{m-1}(x)*(l_{(m-1)1},\ldots,l_{(m-1)(m-2)},a_{m-1},0) \\ =
(e_1a_1+e_2l_{21}+\cdots +e_{m-1}l_{(m-1)1}, e_2a_2+e_3l_{32}+\cdots+e_{m-1}l_{(m-1)2}, \ldots, \\ ~~~~~~~~~~~~~~~~~~~~~~~~~~~~~~~~~~~~~~~~~~~ ~~~~~~~~~~~~ e_{m-2}a_{m-2}+e_{m-1}l_{(m-1)(m-2)}, e_{m-1}a_{m-1},0)$ for some polynomials $e_i(x);~i=1,2,\ldots,m-1.$

So, taking $d_{m-1}=e_{m-1}$,
\begin{equation}
\label{eq3}
d_{m-1}a_{m-1}=\frac{x^{\alpha_m} -1}{a_{m(m-1)}}l_{m(m-1)} \bmod 2^{m-1} 
\end{equation}

Equation (c) together with~\eqref{eq3} proves (iii). 

\item[(iv)] We have $a_{m-1}d_{m-1}=\frac{x^{\alpha_m}-1}{a_{m(m-1)}}l_{m(m-1)}.$
So 
$$\begin{array}{ll}
\beta_1 &=d_{m-1}*(l_{(m-1)1}, l_{(m-1)2}, \ldots, l_{(m-1)(m-2)}, a_{m-1},0) \\ &=(d_{m-1}l_{(m-1)1},d_{m-1}l_{(m-1)2}, \ldots, d_{(m-1)}l_{(m-1)(m-2)}, d_{m-1}a_{m-1},0).
\end{array}$$
Again 
$$\begin{array}{ll}
\beta_2 &=\frac{x^{\alpha_m}-1}{a_{m(m-1)}}*(l_{m1},l_{m2}, \ldots, l_{m(m-1)},a_m) \\ &=
(\frac{x^{\alpha_m}-1}{a_{m(m-1)}}l_{m1},\frac{x^{\alpha_m}-1}{a_{m(m-1)}}l_{m2}, \ldots, \frac{x^{\alpha_m}-1}{a_{m(m-1)}}l_{m(m-1)},0 ) \\ &=
(\frac{x^{\alpha_m}-1}{a_{m(m-1)}}l_{m1},\frac{x^{\alpha_m}-1}{a_{m(m-1)}}l_{m2}, \ldots, \frac{x^{\alpha_m}-1}{a_{m(m-1)}}l_{m(m-2)},d_{m-1} a_{m-1},0).
\end{array}$$
Thus $\beta_1-\beta_2 \in \ker \Psi$. The last two components are zero. So, we may treat $\beta_1-\beta_2$ as a codeword in the $\prod\limits_{i=1}^{m-2}\mathbb{Z}_{2^i}$-additive cyclic code
generated by $$\langle (a_1,0,\ldots,0), (l_{21},a_2,0,\ldots,0),\ldots, (l_{(m-2)1}, l_{(m-2)2}, \ldots, l_{(m-2)(m-3)},a_{m-2}),$$ where each generator is an $(m-2)$-tuple. Hence by (d)
\begin{equation}
\label{eq4}
a_{m-2} | d_{m-1}l_{(m-1)(m-2)}-\frac{x^{\alpha_m}-1}{a_{m(m-1)}}l_{m(m-2)} \bmod 2^{m-2}
\end{equation}
Clearly (d) along with~\eqref{eq4} proves (iv). 

Thus we have our theorem.
\end{enumerate}
%\qed
\end{proof}

We now find a minimal spanning set for a $\prod\limits_{i=1}^n\mathbb{Z}_{2^i}$ additive cyclic code $\mathcal{C}$ generated by 
$\langle (a_1,0, \ldots,0),$ $ (l_{21},a_2,0, \ldots,0), \ldots, (l_{n1},l_{n2}, \ldots, l_{n(n-1)},a_{n})\rangle $
with $a_i=\sum\limits_{j_i=0}^{i-1}2^{j_i}a_{ij_i},$ where $a_{ij}$'s and $l_{ij}$'s are defined as above. First we prove the following lemma.

\begin{lemma}\label{Lm2}
Let $i \in \{1,2,\ldots,n\}$. There exists polynomials $f_{j_i i}; 1\leq j_i \leq i-1$ such that $$l_{i j_i} h_{i(i-1)}= a_{j_i} f_{j_i i} + \sum\limits_{k= j_i +1}^{i-1} l_{k j_i} f_{ki}$$ where $a_{i j_i} h_{i j_i} = x^{\alpha_i}-1;~ 0\leq j_i \leq i-1,2\leq i\leq n.$
\end{lemma}
\begin{proof}
We prove the given theorem by induction on $i$.

From Theorem~\eqref{Th2} we have $l_{21} h_{21}=a_1 f_{11}$ where $f_{11}=d_1.$ So the given statement is true for $i=2$.

Let us assume that the given statement is true for $i=2,3,\ldots,m-1.$ We will prove it to be true for $i=m$.

Consider $j_m = m-1$. Clearly from Theorem~\eqref{Th2} we have $l_{m(m-1)} h_{m(m-1)}=a_{m-1} d_{m-1}$. So, taking $f_{(m-1)m}=d_{m-1}$ we have our statement to be true for $j_m = m-1$. Clearly, 

$ f_{(m-1)m}*(l_{(m-1)1}, l_{(m-1)2},\ldots, l_{(m-1)(m-2)},a_{m-1},0,\ldots,0)\\ ~~~~~~~~~~~~~~~~~~~~~~~~~~~~~~~~~~~~~~~~~~~~~~~~~~~~- h_{m(m-1)}* (l_{m1},l_{m2},\ldots,l_{m(m-1)},0,\ldots,0)\\ = (f_{(m-1)m} l_{(m-1)1}- l_{m1} h_{m(m-1)},\ldots,\\~~~~~~~~~~~~~~~~~~~~~~~~~~~~~~~~~~~~~~~~~~~~~ f_{(m-1)m} l_{(m-1)(m-2)}- l_{m(m-2)} h_{m(m-1)},0,\ldots,0)\in \ker\Psi.$

So, $-a_{m-2} f_{(m-2)m} = f_{(m-1)m} l_{(m-1)(m-2)}- l_{m(m-2)} h_{m(m-1)}$ for some polynomial $f_{(m-2)m}.$\\ i.e. $$ l_{m(m-2)} h_{m(m-1)}= a_{m-2} f_{(m-2)m}  +  l_{(m-1)(m-2)} f_{(m-1)m}.$$ So, for $i=m,~ j_m=m-2$ the statement is true. Proceeding similarly as above after $j_m$ steps we have \\$(l_{m1} h_{m(m-1)} -\sum\limits_{k=j_m-1}^{m-1} l_{k1} f_{km}, l_{m2} h_{m(m-1)} -\sum\limits_{k=j_m-1}^{m-1} l_{k2} f_{km}, \\ ~~~~~~~~~~~~~~~~~~~~~~~~~~~~~~~~~~~~~~~~~~~~~~~~~~~~~ \ldots, l_{m(j_m-1)} h_{m(m-1)} -\sum\limits_{k=j_m-1}^{m-1} l_{k(j_m-1)} f_{km},0\ldots,0)\in \ker \Psi.$ \\ Hence there is a polynomial $f_{(j_m-2)m}$ such that
$$a_{j_m-2} f_{(j_m-2)m} = l_{m(j_m-2)} h_{m(m-1)} - \sum\limits_{k=j_m-1}^{m-1} l_{k(j_m-2)} f_{km},$$  i.e. $$ l_{m(j_m-2)} h_{m(m-1)} = a_{j_m-2} f_{(j_m-2)m}  + \sum\limits_{k=j_m-1}^{m-1} l_{k(j_m-2)} f_{km}.$$ At last we have $(l_{m1} h_{m(m-1)} - \sum\limits_{k=2}^{m-1} l_{k1} f_{km},0,\ldots,0)\in \ker\Psi$ and hence there exists a polynomial $f_{1m}$ such that $a_1 f_{1m} = l_{m1} h_{m(m-1)} - \sum\limits_{k=2}^{m-1} l_{k1} f_{km}$. i.e. $$l_{m1} h_{m(m-1)}= a_1 f_{1m} + \sum\limits_{k=2}^{m-1} l_{k1} f_{km}.$$ So the given statement is true for $i=m$. Thus by principle of induction we have our desired result.
%\qed
\end{proof}

\begin{theorem}
\label{th3}
Let $$\mathcal{C}=\langle(a_1,0,\ldots,0), (l_{21},a_2,0,\ldots,0), (l_{31},l_{32},a_3,0,\ldots,0), \ldots, (l_{n1},l_{n2}, \ldots, l_{n(n-1)},a_n) \rangle$$ where $a_i=\sum\limits_{j_i=0}^{i-1} 2^{j_i}a_{ij_i}$ satisfying 
$a_{ij_i}h_{ij_i}=x^{\alpha_i}-1; a_{ij_i}m_{ij_i}=a_{i(j_i-1)}; 0\leq j_i \leq i-1$. 
Let $$\mathcal{S}_{i0}=\bigcup\limits_{k=0}^{\deg h_{i0}-1}\{x^k*(l_{i1},l_{i2},\ldots, l_{i(i-1)},a_i,0,\ldots,0) \},$$
    $$\mathcal{S}_{ij_i}=\bigcup\limits_{k=0}^{\deg m_{ij_i}-1}\bigg\{x^k*(l_{i1}h_{i(j_i-1)},l_{i2}h_{i(j_i-1)},\ldots, l_{i(i-1)}h_{i(j_i-1)},
            \sum_{p=j_i}^{i-1}a_{ip}h_{i(j_i-1)},0,\ldots,0) \bigg\},$$
$0 \leq j_i \leq i-1, 1 \leq i \leq n.$
Then $\mathcal{S}=\bigcup\limits_{i=1}^n (\bigcup\limits_{j_i=0}^{i-1}\mathcal{S}_{ij_i})$ forms a minimal spanning set for the $\prod\limits_{i=1}^n \mathbb{Z}_{2^i}$-additive cyclic code $ \mathcal{C}$ and $\mathcal{C}$ has $2^{\sum\limits_{i=1}^n ( i \deg h_{i0} + \sum\limits_{j_i=1}^{i-1}(i-j_i)\deg m_{ij_i})}$ codewords. For any codeword $c$ there are polynomials  $e_{i j_i}, ~0 \leq j_i \leq i-1,~1\leq i\leq n$ satisfying $\deg e_{i0} \leq \deg h_{i0}-1$ and $\deg e_{i j_i} \leq \deg m_{i j_i}-1 ~ \forall ~ 1 \leq j_i \leq i-1, ~1\leq i \leq n $ such that $c=\sum\limits_{i=1}^n (e_{i0}*(l_{i1},l_{i2},\ldots, l_{i(i-1)},a_i,0,\ldots,0) \\~~~~~~ + \sum\limits_{j_i=1}^{i-1} e_{i j_i}*(l_{i1}h_{i(j_i-1)},l_{i2}h_{i(j_i-1)},\ldots, l_{i(i-1)}h_{i(j_i-1)},
            \sum\limits_{p=j_i}^{i-1}a_{ip}h_{i(j_i-1)},0,\ldots,0)).$
\end{theorem}

\begin{proof}
Let $c=c(x) \in \mathcal{C}$, then $\exists ~ d_i \in \mathbb{Z}_{2^n}[x]$ for $i=1,2,\ldots,n$ such that \\ 
$$\begin{array}{ll}
c &=d_1*(a_1,0,\ldots,0)+d_2 *(l_{21},a_2,0,\ldots,0)+ \cdots+d_n*(l_{n1},l_{n2}, \ldots, l_{n(n-1)},a_n) \\ &=\sum\limits_{i=1}^n d_i*(l_{i1},l_{i2}, \ldots, l_{i(i-1)},a_i,0,\ldots, 0).
\end{array}$$
If $\deg d_1 \leq \deg h_{10}-1$, then $d_1*(a_1,0,\ldots,0) \in Span (\mathcal{S}_{10})$. Otherwise by division algorithm, there exist
polynomials $q_{11}(x), r_{11}(x)$ such that $d_1=h_{10}a_{11}+r_{11}$ where either $r_{11}=0$ or $\deg r_{11} < \deg h_{10}$.
Then $d_1*(a_1,0, \ldots, 0)=r_{11}*(a_1,0,\ldots,0) \in Span(\mathcal{S}_{10}).$
Now consider $d_i;~ i \in\{2,3,$ $\ldots,n\}$.
Clearly if $\deg d_i \leq \deg h_{i0}-1$, then 
$d_i*(l_{i1},l_{i2},\ldots,l_{i(i-1)},a_i,0,\ldots,0)$ $\in Span(\mathcal{S}_{i0})$. Otherwise by division algorithm
there are polynomials $q_{i1}(x), r_{i1}(x)$ such that $d_i=q_{i1}h_{i0}+r_{i1}$ where either $r_{i1}=0$ or 
$0 \leq \deg r_{i1}<\deg h_{i0}$. So, \\
$~~~~~d_i*(l_{i1},l_{i2},\ldots, l_{i(i-1)},a_i,0,\ldots,0) \\ =q_{i1}*(l_{i1}h_{i0}, l_{i2}h_{i0}, \ldots, l_{i(i-1)}h_{i0},\sum\limits_{p=1}^{i-1}2^pa_{ip}h_{i0},0,\ldots,0) \\ ~~~~~~~~~~~~~~~~~~~~~~~~~~~~~~~~~~~~~~~~~~~~~~~~~~~~~+r_{i1}*(l_{i1},l_{i2}, \ldots, l_{i(i-1)},a_i,0,\ldots,0).$

Clearly $r_{i1}*(l_{i1},l_{i2}, \ldots, l_{i(i-1)},a_i,0,\ldots,0) \in Span(\mathcal{S}_{i0}).$
If $\deg q_{i1}<\deg m_{i1},$ then $q_{i1}*(l_{i1}h_{i0},l_{i2}h_{i0}, \ldots, l_{i(i-1)}h_{i0}, \sum\limits_{p=1}^{i-1}2^pa_{ip}h_{i0}, 0,\ldots,0)
\in Span(\mathcal{S}_{i1}).$ Otherwise by division algorithm there are polynomials $q_{i2}(x), r_{i2}(x)$ such that 
$q_{i1}=q_{i2}m_{i1}+r_{i2}$ where $r_{i2}=0$ or $0 \leq r_{i2}< \deg m_{i1}$.

Now $a_{ij_i}m_{ij_i}=a_{i(j_i-1)}  \ \forall ~ 1 \leq j_i \leq i-1.$ So, 
$a_{ij_i}m_{ij_i}h_{i(j_i-1)}=a_{i(j_i-1)}h_{i(j_i-1)}=x^{\alpha_i}-1 =a_{ij_i}h_{ij_i}$.

$\Rightarrow m_{ij_i}h_{i(j_i-1)}=h_{ij_i}.$

So, $ q_{i1}*(l_{i1}h_{i0},l_{i2}h_{i0}, \ldots, l_{i(i-1)}h_{i0}, \sum\limits_{p=1}^{i-1}2^pa_{ip}h_{i0},0,\ldots,0)\\=
q_{i2}*(l_{i1}h_{i1},l_{i2}h_{i1}, \ldots, l_{i(i-1)}h_{i1}, \sum\limits_{p=2}^{i-1}2^pa_{ip}h_{i1},0,\ldots,0)+ \\ ~~~~~~~~~~~~~~~~~~~~~~~~~~~~~~~~r_{i2}*(l_{i1}h_{i0},l_{i2}h_{i0}, \ldots, l_{i(i-1)}h_{i0}, \sum\limits_{p=1}^{i-1}2^pa_{ip}h_{i0},0,\ldots,0).$

Clearly $r_{i2}*(l_{i1}h_{i0},l_{i2}h_{i0}, \ldots, l_{i(i-1)}h_{i0}, \sum\limits_{p=1}^{i-1}2^pa_{ip}h_{i0},0,\ldots,0) \in Span(\mathcal{S}_{i1})$.

Now if $\deg q_{i2}< \deg m_{i2},$ then \\
$q_{i2}*(l_{i1}h_{i1},l_{i2}h_{i1}, \ldots, l_{i(i-1)}h_{i1}, \sum\limits_{p=2}^{i-1}2^pa_{ip}h_{i1},0,\ldots,0) \in Span(\mathcal{S}_{i2}).$
Otherwise there are polynomials $q_{i3}(x), r_{i3}(x)$ such that $q_{i2}=q_{i3}m_{i2}+r_{i3}$ where either $r_{i3}=0$ or 
$0 \leq \deg r_{i3}< \deg m_{12}.$ Proceedings similarly as above if we find $j_i \in \{1,2,\ldots,i-2\}$ such that 
$$q_{ij_i}*(l_{i1}h_{i(j_i-1)}, l_{i2}h_{i(j_i-1)}, \ldots, l_{i(i-1)}h_{i(j_i-1)}, \sum\limits_{p=j_i}^{i-1} 2^p a_{ip}h_{i(j_i-1)},0,\ldots,0)
\in Span(\mathcal{S}_{ij_i})$$
then we are done. Otherwise we must prove that \\
$ q_{i(i-1)}*(l_{i1}h_{i(i-2)}, l_{i2}h_{i(i-2)}, \ldots, l_{i(i-1)}h_{i(i-2)}, 2^{i-1}a_{i(i-1)}h_{i(i-2)},0,\ldots,0) \in Span(\mathcal{S}).$

If $\deg q_{i(i-1)}< \deg m_{i(i-1)}$ then \\
 $ ~~~~~~~~~~~~ q_{i(i-1)}*(l_{i1}h_{i(i-2)}, l_{i2}h_{i(i-2)}, \ldots, l_{i(i-1)}h_{i(i-2)}, \\ ~~~~~~~~~~~~~~~~~~~~~~~~~~~~~~~~~~~~~~~~~~~ 2^{i-1}a_{i(i-1)}h_{i(i-2)},0,\ldots,0) \in Span(\mathcal{S}_{i(i-1)}).$ \\
Otherwise by division algorithm there are polynomials $q_{ii}(x), r_{ii}(x)$ such that $$q_{i(i-1)}=q_{ii}m_{i(i-1)} +r_{ii}$$ where either 
$r_{ii}=0$ or $0 \leq \deg r_{ii}<\deg m_{i(i-1)}.$ Then \\ $~~~~q_{i(i-1)}*(l_{i1}h_{i(i-2)},\ldots, l_{i(i-1)}h_{i(i-2)}, 2^{i-1}a_{i(i-1)}h_{i(i-2)}, 0,\ldots, 0 )\\ = q_{ii}*(l_{i1}h_{i(i-1)},l_{i2}h_{i(i-1)}, \ldots, l_{i(i-1)}h_{i(i-1)}, 0, \ldots,0) \\ ~~~~~~~~~~~~~~~~~~~~~~~~~~~~~+r_{ii}*(l_{i1}h_{i(i-2)}, 
\ldots, l_{i(i-1)}h_{i(i-2)}, 2^{i-1}a_{i(i-1)}h_{i(i-2)},0, \ldots,0).$ \\
Clearly $r_{ii}*(l_{i1}h_{i(i-2)}, 
\ldots, l_{i(i-1)}h_{i(i-2)}, 2^{i-1}a_{i(i-1)}h_{i(i-2)},0, \ldots,0) \in Span(\mathcal{S}_{i(i-1)})$. \\
So we need only to prove that $$q_{ii}*(l_{i1}h_{i(i-1)},l_{i2}h_{i(i-1)}, \ldots, l_{i(i-1)}h_{i(i-1)}, 0, \ldots,0) \in Span(\mathcal{S}).$$

Now we know from the previous Lemma~\eqref{Lm2} that $$l_{i j_i} h_{i(i-1)}= a_{j_i} f_{j_i i} + \sum\limits_{k= j_i +1}^{i-1} l_{k j_i} f_{ki};~\forall~1\leq j_i\leq i-1,~1\leq i\leq n.$$
So, $q_{ii}*(l_{i1}h_{i(i-1)},l_{i2}h_{i(i-1)}, \ldots, l_{i(i-1)}h_{i(i-1)}, 0, \ldots,0) \\ = q_{ii}*(a_{1} f_{1i} + \sum\limits_{k= 2}^{i-1} l_{k 1} f_{ki}, a_{2} f_{2i} + \sum\limits_{k= 3}^{i-1} l_{k 2} f_{ki},\ldots, a_{i-1}f_{(i-1)i},0,\ldots,0) \\ = q_{ii}f_{1i} * (a_1,0,\ldots,0) + q_{ii}f_{2i} * (l_{21},a_2,0,\ldots,0) + q_{ii}f_{3i} * (l_{31},l_{32},a_3,0,\ldots,0) + \cdots + q_{ii}f_{(i-1)i} * (l_{(i-1)1},\ldots,\l_{(i-2)(i-2)},a_{i-1},0,\ldots,0) \in Span(\displaystyle{\bigcup\limits_{k=0}^{i-1} \mathcal{S}_{k 0}}).$

Hence $$d_i * (l_{i1}, l_{i2}, \ldots, l_{i(i-1)},a_i,0,\ldots,0) \in Span( \mathcal{S}).$$
Thus $c(x) \in Span(\bigcup\limits_{i=1}^{n}(\bigcup\limits_{j_i=0}^{i-1} \mathcal{S}_{i j_i})) = Span (\mathcal{S})$. So, $\mathcal{S}$ is a spanning set for $\mathcal{C}.$ Moreover $\mathcal{S}$ is minimal in the sense that no element in $\mathcal{S}$ can be expressed as a linear combination of other elements in $\mathcal{S}.$ Thus $\mathcal{S}$ is a minimal spanning set for $\mathcal{C}.$

Now clearly for $i=\{1,2,\ldots,n\}$, $\mathcal{S}_{i0}$ will contribute $2^{i \deg h{i0}}$ codewords and $\mathcal{S}_{i j_i}$ will contribute $2^{(i-j_i) \deg m_{i j_i}} $ codewords for each $j_i \in \{ 1,2,\ldots,i-1\}.$ So, the total number of codewords in $\mathcal{C}$ is $\prod\limits_{i=1}^n(2^{i \deg h_{i0}} \prod\limits_{j_i=1}^{i-1} 2^{(i-j_i)\deg m_{i j_i}})= 2^{\sum\limits_{i=1}^n ( i \deg h_{i0} + \sum\limits_{j_i=1}^{i-1}(i-j_i)\deg m_{ij_i})}.$ 
%\qed
\end{proof}
\begin{corollary}
Let $m\in \{1,2,\ldots,n \}$ and $$\mathcal{C}=\langle(l_{m1},l_{m2}, \ldots, l_{m(m-1)},a_m,0,\ldots,0) \rangle$$ where $a_m=\sum\limits_{j_m=0}^{m-1} 2^{j_m}a_{m j_m}$ satisfying 
$a_{mj_m}h_{m j_m}=x^{\alpha_m}-1; a_{m j_m}m_{m j_m}=a_{m(j_m-1)}; 0\leq j_m \leq m-1$. 
Let $$\mathcal{S}_{m0}=\bigcup\limits_{k=0}^{\deg(h_{m0})-1}\{x^k*(l_{m1},l_{m2},\ldots, l_{m(m-1)},a_m,0,\ldots,0) \},$$ $\mathcal{S}_{m j_m}=\bigcup\limits_{k=0}^{\deg(m_{mj_m})-1}\bigg\{x^k*(l_{m1}h_{m(j_m-1)},l_{m2}h_{m(j_m-1)},\ldots, l_{m(m-1)}h_{m(j_m-1)}, \\ ~~~~~~~~~~~~~~~~~~~~~~~~~~~~ ~~~~~~~~~~~~~~~~~~~~~~~~~~~~~~~ ~~~~~~
            \sum\limits_{p=j_m}^{m-1}a_{mp}h_{m(j_m-1)},0,\ldots,0) \bigg\},$ \\
$0 \leq j_m \leq m-1.$
Then $\mathcal{S}=\bigcup\limits_{j_m=0}^{m-1}\mathcal{S}_{m j_m}$ forms a minimal spanning set for the $\prod\limits_{i=1}^n \mathbb{Z}_{2^i}$-additive cyclic code $ \mathcal{C}$ and $\mathcal{C}$ has $2^{ m \deg h_{m0} + \sum\limits_{j_m=1}^{m-1}(m-j_m)\deg m_{m j_m}}$ codewords. Any codeword $c$ can be written as \\ $c=(e_{m0}*(l_{m1},l_{m2},\ldots, l_{m(m-1)},a_m,0,\ldots,0) \\~~~~~~ + \sum\limits_{j_m=1}^{m-1} e_{m j_m}*(l_{m1}h_{m(j_m-1)},l_{m2}h_{m(j_m-1)},\ldots, l_{m(m-1)}h_{m(j_m-1)}, \\ ~~~~~~~~~~~~~~~~~~~~~~~~~~~~~~~ ~~~~~~~~~~~~~~~~~~~~~~~~~~
            \sum\limits_{p=j_m}^{m-1}a_{mp}h_{m(j_m-1)},0,\ldots,0)),$ \\
             $e_{m j_m}$'s, $0 \leq j_m \leq m-1$, being suitable polynomials satisfying $\deg e_{m0} \leq \deg h_{m0}-1$ and $\deg e_{m j_m} \leq \deg m_{m j_m}-1 ~ \forall ~ 1 \leq j_m \leq m-1.$
\end{corollary}
\begin{proof}

Taking $a_{i j_i} = x^{\alpha_i}-1 $ in the previous theorem gives $\mathcal{S}_{i j_i} =\phi ~\forall~ 0\leq j_i \leq i-1, i=1,2,\ldots,m-1,m+1,\ldots,n $ and hence the corollary follows.

%\qed
\end{proof}
\begin{example}
Let $\mathcal{C}$ be a $\mathbb{Z}_2 \mathbb{Z}_4 \mathbb{Z}_8$-additive cyclic code in 
$$\mathcal{R}_{8,5,5}=\mathbb{Z}_2[x]/\langle x^8-1\rangle \times \mathbb{Z}_4[x]/\langle x^5-1\rangle \times\mathbb{Z}_8[x]/\langle x^5-1\rangle  $$
generated by $\big((a_{10},0,0), (l_{21},a_{20}+2a_{21},0), (l_{31},l_{32}, a_{30}+2a_{31}+4a_{32}) \big),$
where 
\begin{itemize}
 \item[~] $a_{10}=1+x^2, l_{21}=l_{31}=1+x, $
  \item[~]    $a_{20}=2x^2+3, a_{21}=3, l_{32}=3x $
   \item[~]   $ a_{30}=2x+3, a_{31}=3, a_{32}=2x^2+3.$
\end{itemize}
Hence the generator matrix can be obtained from Theorem~\ref{th3} as follows:

\[
\left(\begin{array}{cccccccccccccccccccc}
1 &~ 0 &~ 1 &~ 0 &~ 0 &~ 0 &~ 0 &~0 &~0 &~0 &~0 &~0 &~0 &~0 &~0 &~0 &~0 &~0 &~ \\
0 &~1 &~ 0 &~ 1 &~ 0 &~ 0 &~ 0 &~ 0 &~0 &~0 &~0 &~0 &~0 &~0 &~0 &~0 &~0 &~0  &~\\
0 &~0 &~ 1 &~ 0 &~ 1 &~ 0 &~ 0 &~ 0 &~ 0 &~0 &~0 &~0 &~0 &~0 &~0 &~0 &~0 &~0  &~\\
0 &~0 &~ 0 &~ 1 &~ 0 &~ 1 &~ 0 &~ 0 &~ 0 &~ 0 &~0 &~0 &~0 &~0 &~0 &~0 &~0 &~0 &~\\
0 &~0 &~0 &~ 0 &~ 1 &~ 0 &~ 1 &~ 0 &~ 0 &~ 0 &~ 0 &~0 &~0 &~0 &~0 &~0 &~0 &~0 &~\\
0 &~0 &~0 &~ 0 &~ 1 &~ 0 &~ 1 &~ 0 &~ 0 &~ 0 &~ 0 &~0 &~0 &~0 &~0 &~0 &~0 &~0 &~\\
%\hline 
1 &~1 &~0 &~ 0 &~ 0 &~ 0 &~ 0 &~0  &~ 1 &~ 0 &~2  &~0 &~0 &~0 &~0 &~0 &~0 &~0 &~\\
0 &~ 1 &~1 &~0 &~ 0 &~ 0 &~ 0 &~ 0 &~0  &~ 1 &~ 0 &~2  &~0 &~0 &~0 &~0 &~0 &~0 &~\\
0 &~0 &~ 1 &~1 &~0 &~ 0 &~ 0 &~ 0 &~ 0 &~0  &~ 1 &~ 0 &~2  &~0 &~0 &~0 &~0 &~0 &~\\
%\hline 
1 &~1 &~0 &~ 0 &~ 0 &~ 0 &~ 0 &~0  &~ 0 &~ 3 &~0  &~0 &~0 &~5 &~2 &~0 &~0 &~0 &~\\
0 &~ 1 &~1 &~0 &~ 0 &~ 0 &~ 0 &~ 0 &~0  &~ 0 &~ 3 &~0  &~0 &~0 &~5 &~2 &~0 &~0 &~\\
0 &~0 &~ 1 &~1 &~0 &~ 0 &~ 0 &~ 0 &~ 0 &~0  &~ 0 &~ 3 &~0  &~0 &~0 &~5 &~2 &~0 &~\\
0 &~ 0 &~0 &~ 1 &~1 &~0 &~ 0 &~ 0 &~ 0 &~ 0 &~0  &~ 0 &~ 3 &~0  &~0 &~0 &~5 &~2 &~\\
%\hline 
1 &~1 &~0 &~ 0 &~ 0 &~ 1 &~ 1 &~0  &~ 0 &~ 0 &~0  &~0 &~0 &~0 &~0 &~0 &~0 &~0 &~\\
\end{array}\right)
.
\] 

\end{example}
\section{Conclusion}
In this paper we have used basic results and techniques from algebra \cite{Hungerford74,McDonald74} and coding theory \cite{Ling}. We have introduced $\prod\limits_{i=1}^n \mathbb{Z}_{2^i}$-additive cyclic codes. We have defined a $\prod\limits_{i=1}^n \mathbb{Z}_{2^i}$-additive cyclic code of a certain length and a notion of distance in it has been introduced using the generalised Gray maps. We have showed that the dual of such a code is also cyclic. Then we have identified these codes as $\mathbb{Z}_{2^n}[x]$-submodules of $\prod\limits_{i=1}^n \mathbb{Z}_{2^i}[x] / \langle x^{\alpha_i}-1 \rangle$; $\alpha_i$ and $i$ being relatively prime for each $i=1,2,\ldots,n.$ We gave the polynomial definition of a $\prod\limits_{i=1}^n \mathbb{Z}_{2^i}$-additive cyclic code. We have studied the structure of these codes and determined a set of generators for the same. We also have determined a minimal spanning set for such a code and the total number of codewords in it. Finally we present an illustrative example. 

Since this family of codes are new, many things are left to be explored. One can try to obtain self-dual codes of this type. Also perfect codes may be of great interest. This family contains all binary codes and also $\mathbb{Z}_{2^r} \mathbb{Z}_{2^s}$-additive cyclic codes for $r<s \in \mathbb{N}$. Moreover the idea will work same for obtaining $\prod\limits_{i=1}^n \mathbb{Z}_{p^i}$-additive cyclic codes; $p$ being a prime. So this family is one of the most generalised family of codes.
%\section*{References}

\end{document}